\documentclass[journal]{IEEEtran}

\usepackage{graphicx,url}
\usepackage{dcolumn}
\usepackage{bm}
\usepackage{amsmath,amsfonts,amssymb,citesort}

\newtheorem{theorem}{\noindent{\it Theorem}}[section]

\newtheorem{lemma}[theorem]{\noindent{\it Lemma}}

\newtheorem{remark}[theorem]{\noindent{\it Remark}}

\newenvironment{proof}{\noindent{\it Proof:}}{$\hfill$ $\Box$\\ }

\begin{document}

\title{Asymmetric Quantum Codes: New Codes from Old}

\author{Giuliano G. La Guardia
\thanks{Giuliano G. La Guardia is with Department of Mathematics and Statistics,
State University of Ponta Grossa, 84030-900, Ponta Grossa - PR, Brazil.
E-mail:~{\tt \small gguardia@uepg.br}.
}}

\maketitle

\begin{abstract}
In this paper we extend to asymmetric quantum error-correcting
codes (AQECC) the construction methods, namely: puncturing,
extending, expanding, direct sum and the $({ \bf u}| { \bf u}+{
\bf v})$ construction. By applying these methods, several families
of asymmetric quantum codes can be constructed. Consequently, as
an example of application of quantum code expansion developed
here, new families of asymmetric quantum codes derived from
generalized Reed-Muller (GRM) codes, quadratic residue (QR),
Bose-Chaudhuri-Hocquenghem (BCH), character codes and
affine-invariant codes are constructed.
\end{abstract}

\section{Introduction}

To make reliable the transmission or storage of quantum
information against noise caused by the environment there exist
many works available in the literature dealing with constructions
of efficient quantum error-correcting codes (QECC) over unbiased
quantum channels
\cite{Calderbank:1998,Rains:1999,Steane:1999,Ashikhmin:2000,Ashikhmin:2001,Grassl:2003,Ketkar:2006,LaGuardia:2009,LaGuardia:2011A}.
Recently, these constructions have been extended to asymmetric
quantum channels in a natural way
\cite{Ioffe:2007,Evans:2007,Stephens:2008,Salah:2008,Sarvepalli:2008,Sarvepalli:2009,Salah:2010,Ezerman:2010,Wang:2010,Ezerman:2010A,Ezerman:2010B,LaGuardia:2011,LaGuardia:2012,LaGuardia:2012I}.

Asymmetric quantum error-correcting codes (AQECC) are quantum
codes defined over quantum channels where qudit-flip errors and
phase-shift errors may have different probabilities. Steane
\cite{Steane:1996} was the first author who introduced the notion
of asymmetric quantum errors. As usual, the parameters ${[[n, k, \
d_{z}/d_{x}]]}_{q}$ denote an asymmetric quantum code, where
$d_{z}$ is the minimum distance corresponding to phase-shift
errors and $d_{x}$ is the minimum distance corresponding to
qudit-flip errors. The combined amplitude damping and dephasing
channel (specific to binary systems; see \cite{Sarvepalli:2008})
is an example for a quantum channel that satisfies $d_{z} >
d_{x}$, i. e., the probability of occurrence of phase-shift errors
is greater than the probability of occurrence of qudit-flip
errors.

Let us give a brief summary of the papers available in the
literature dealing with AQECC. In \cite{Evans:2007}, the authors
explored the asymmetry between qubit-flip and phase-shift errors
to perform an optimization when compared to QECC. In
\cite{Ioffe:2007} the authors utilize BCH codes to correct
qubit-flip errors and LDPC codes to correct more frequently
phase-shift errors. In \cite{Stephens:2008} the authors consider
the investigation of AQECC via code conversion. In the papers
\cite{Salah:2008,LaGuardia:2011}, families of AQECC derived from
BCH codes were constructed. Asymmetric stabilizer codes derived
from LDPC codes were constructed in \cite{Sarvepalli:2008}, and in
\cite{Sarvepalli:2009}, the same authors have constructed several
families of both binary and nonbinary AQECC as well as to derive
bounds such as the (quantum) Singleton and the linear programming
bound to AQECC. In \cite{Salah:2010}, both AQECC (derived from
cyclic codes) and subsystem codes were investigated. In
\cite{Wang:2010}, the construction of nonadditive AQECC as well as
constructions of asymptotically good AQECC derived from
algebraic-geometry codes were presented. In \cite{Ezerman:2010},
the Calderbank-Shor-Steane (CSS) construction
\cite{Nielsen:2000,Calderbank:1998,Ketkar:2006} was extended to
include codes endowed with the Hermitian and also trace Hermitian
inner product. In \cite{Ezerman:2010A}, asymmetric quantum MDS
codes derived from generalized Reed-Solomon (GRS) codes were
constructed. More recently, in
\cite{LaGuardia:2012,LaGuardia:2012I}, constructions of families
of AQECC by expanding GRS codes and by applying product codes,
respectively, were presented.

In this paper we extend to asymmetric quantum error-correcting
codes (AQECC) the construction methods, namely: puncturing,
extending, expanding, direct sum and the $({ \bf u}| { \bf u}+{
\bf v})$ construction. An interesting fact pointed out by the
referee is that the results presented in the first version of this
paper (constructions of asymmetric quantum codes derived from
classical linear codes endowed with the Euclidean as well as with
the Hermitian inner product) also hold in a more general setting,
i. e., constructions of asymmetric quantum codes derived from
additive codes (see \cite{Calderbank:1998,Ketkar:2006}, where in
\cite{Calderbank:1998} a general theory of quantum codes over
$GF(4)$ was developed, and in \cite{Ketkar:2006} a generalization
to nonbinary alphabets was presented). Because of this fact, we
keep the original constructions of AQECC derived from linear codes
and we also add more results with respect to constructions of
AQECC derived from additive codes. More specifically, concerning
the techniques of extending and the $({ \bf u}| { \bf u}+{ \bf
v})$ construction, the arguments shown in this paper to AQECC
codes derived from linear codes are similar to the ones derived
from additive codes. The techniques of puncturing, expanding and
direct sum will be shown in two different ways (each of them), i.
e., AQECC derived from linear and additive codes. We keep both
styles of constructions (additive/linear) in this paper because
although the first (additive) is more general, we utilize
different tools to show the results for the linear case, and these
tools can be applied in future works.

The paper is organized as follows. In Section~\ref{sec2} we fix
the notation. In Section~\ref{sec3} we recall the concepts and
definitions of AQECC and error operators. Section~\ref{sec4} is
devoted to establish the construction methods. More precisely, we
show how to construct new AQECC by means of the techniques of
puncturing, extending, expanding, direct sum and the $({ \bf u}| {
\bf u}+{ \bf v})$ construction. In Section~\ref{sec5}, we utilize
the quantum code expansion developed in Section~\ref{sec4} applied
to (classical) generalized Reed-Muller (GRM) codes, quadratic
residue, character codes, BCH and affine-invariant codes in order
to construct several new families of AQECC. Finally, in
Section~\ref{sec6}, we discuss the contributions presented in this
paper.

\section{Notation}\label{sec2}

Throughout this paper, $p$ denotes a prime number, $q$ denotes a
prime power, ${\mathbb F}_{q}$ is a finite field with $q$
elements, $\alpha \in {\mathbb F}_{q^m}$ is a primitive $n$ root
of unity. The (Hamming) distance of two vectors ${ \bf v}, { \bf w
}\in {\mathbb F}_{q}^{n}$ is the number of coordinates in which ${
\bf v}$ and ${ \bf w }$ differ. The (Hamming) weight of a vector
${ \bf v}=(v_1, v_2, \ldots, v_{n}) \in {\mathbb F}_{q}^{n}$ is
the number of nonzero coordinates of ${ \bf v}$. The trace map
tr$_{q^{m}/q}: {\mathbb F}_{q^{m}} \longrightarrow {\mathbb
F}_{q}$ is defined as tr$_{q^{m}/q}(a):= \displaystyle
\sum_{i=0}^{m-1} a^{q^{i}}$. We denote $H\leq G$ to mean that $H$
is a subgroup of a group $G$; the center of $G$ is denoted by
$Z(G)$. If $S\leq G$ then we denote by $C_{G}(S)$ the centralizer
of $S$ in $G$; $SZ(G)$ denotes the subgroup generated by $S$ and
the center $Z(G)$.

As usual, ${[n, k, d]}_{q}$ denotes the parameters of a classical
linear code $C$ over ${\mathbb F}_{q}$, of length $n$, dimension
$k$ and minimum distance $d$. We denote by wt$(C)$ the minimum
weight of $C$, and by d$(C)$ the minimum distance of $C$.
Sometimes we have abused the notation by writing $C = {[n, k,
d]}_{q}$. If $C$ is an ${[n, k, d]}_{q}$ code then its Euclidean
dual is defined as ${C}^{{\perp}} = \{ {\bf y} \in \displaystyle
{\mathbb F}_{q}^{n} \mid {\bf y} \cdot {\bf x} = 0, \forall \ {\bf
x} \in C\}$; in the case that $C$ is an ${[n, k, d]}_{q^{2}}$
code, then its Hermitian dual is defined by ${C}^{{\perp}_{h}} =
\{ {\bf y} \in \displaystyle {\mathbb F}_{q^{2}}^{n} \mid {\bf
y}^q \cdot {\bf x} = 0, \forall \ {\bf x} \in C\}$, where ${\bf
y}^q = (\displaystyle y_{1}^{q}, \ldots, y_{n}^{q} )$ denotes the
conjugate of the vector ${\bf y} = (y_1, \ldots, y_n )$. If ${ \bf
a}=(a_1, \ldots , a_{n})$ and ${ \bf b}=(b_1, \ldots , b_{n})$ are
two vectors in ${\mathbb F}_{q}^{n}$ then the symplectic weight
swt of the vector $({ \bf a}|{ \bf b}) \in {\mathbb F}_{q}^{2n}$
is defined by swt$(({ \bf a}|{ \bf b})) = \ \# \{i: 1\leq i\leq n
| (a_i , b_i)\neq (0, 0)\}$. The trace-symplectic form of two
vectors $({ \bf a}|{ \bf b}), ({ \bf {a}^{*}}|{ \bf {b}^{*}}) \in
{\mathbb F}_{q}^{2n}$ is defined by ${\langle ({ \bf a}|{ \bf b})|
({ \bf {a}^{*}}|{ \bf {b}^{*}})\rangle}_{s}=$ tr$_{q/p}({ \bf
b}\cdot { \bf {a}^{*}}-{ \bf {b}^{*}}\cdot{ \bf a})$. If $C\leq
{\mathbb F}_{q}^{2n}$ is an additive code then swt$(C)$ denotes
the symplectic weight of $C$ and $C^{{\perp}_{s}}$ denotes the
trace-symplectic dual of $C$. Similarly, if $C \leq {\mathbb
F}_{q^{2}}^{n}$ is an additive code then $C^{{\perp}_{a}}$ denotes
the trace-alternating dual of $C$, where the trace-alternating
form of two vectors $ { \bf v}, { \bf w} \in {\mathbb
F}_{q^{2}}^{n}$ is defined as ${\langle { \bf v} | { \bf w}
\rangle}_{a}=$ tr$_{q/p}\left( \frac{{ \bf v}\cdot{ \bf w}^{q}-{
\bf v}^{q}\cdot{ \bf w}}{{\beta}^{2q}-{\beta}^{2}}\right)$, where
$( \beta, {\beta}^{q})$ is a normal basis of ${\mathbb F}_{q}^{2}$
over ${\mathbb F}_{q}$.

\section{Error Groups and Asymmetric Codes}\label{sec3}

In this section we recall some basic concepts on quantum error
operators \cite{Calderbank:1998,Ketkar:2006,Sarvepalli:2009} and
asymmetric quantum codes.

Let $\cal{H}$ be the Hilbert space ${\cal H} = {\mathbb C}^{q^n} =
{\mathbb{C}}^{q} \otimes \ldots \otimes {\mathbb{C}}^{q}$. Let
$\mid$$x \rangle$ be the vectors of an orthonormal basis of
${\mathbb{C}}^{q}$, where the labels $x$ are elements of ${\mathbb
F}_{q}$. Consider $a, b \in {\mathbb F}_{q}$; the unitary
operators $X(a)$ and $Z(b)$ on ${\mathbb{C}}^{q}$ are defined by
$X(a)$$\mid$$x \rangle =$$\mid$$x + a\rangle$ and $Z(b)$$\mid$$x
\rangle = w^{tr_{q/p}(bx)}$$\mid$$x\rangle$, respectively, where
$w=\exp (2\pi i/ p)$ is a $p$th root of unity.

Consider that ${\bf a}= (a_1, \ldots , a_n) \in {\mathbb
F}_{q}^{n}$ and ${\bf b}= (b_1, \ldots , b_n) \in {\mathbb
F}_{q}^{n}$. Denote by $X({\bf a})= X(a_1)\otimes \ldots \otimes
X(a_n)$ and $Z({\bf b})= Z(b_1)\otimes \ldots \otimes Z(b_n)$ the
tensor products of $n$ error operators. The set ${\bf E}_{n} = \{
X({\bf a})Z({\bf b}) \mid {\bf a}, {\bf b} \in {\mathbb F}_{q}^{n}
\}$ is an \emph{error basis} on the complex vector space ${\mathbb
C}^{q^n}$ and the set $ {\bf G}_{n} = \{w^c X({\bf a})Z({\bf b})
\mid {\bf a}, {\bf b} \in {\mathbb F}_{q}^{n} , c\in {\mathbb F}_p
\}$ is the \emph{error group} associated with ${\bf E}_{n}$. For a
quantum error $e = w^c  X({\bf a})Z({\bf b}) \in {\bf G}_{n}$ the
$X$-weight is given by wt$_{X}(e) = \ \# \{i: 1\leq i\leq n | a_i
\neq 0\}$; the $Z$-weight is defined as wt$_{Z}(e) = \ \# \{i:
1\leq i\leq n | b_i \neq 0\}$ and the symplectic (or quantum)
weight swt$(e) = \ \# \{i: 1\leq i\leq n | (a_i , b_i)\neq (0,
0)\}$. An AQECC with parameters ${((n, K,d_{z}/ d_{x}))}_{q}$ is
an $K$-dimensional subspace of the Hilbert space ${\mathbb
C}^{q^n}$ and corrects all qudit-flip errors up to $\lfloor
\frac{d_{x}-1}{2} \rfloor$ and all phase-shift errors up to
$\lfloor \frac{d_{z}-1}{2} \rfloor$. An ${((n, q^{k}, d_{z}/
d_{x}))}_{q}$ code is denoted by ${[[n, k, d_{z}/ d_{x}]]}_{q}$.

Let us recall the well-known CSS construction:

\begin{lemma}\cite{Ketkar:2006,Calderbank:1998,Nielsen:2000}(CSS
construction)\label{CSS1} Let $C_1$ and $C_2$ denote two classical
linear codes with parameters ${[n, k_1, d_1]}_{q}$ and ${[n, k_2,
d_2]}_{q}$, respectively. Assume that $C_2\subset C_1$. Then there
exists an AQECC with parameters $[[n, K = k_1-k_2,
d_{z}/d_{x}]{]}_{q}$, where $d_{x}=$wt$(C_2^{\perp} \backslash
C_1^{\perp}) \}$ and $d_{z}=$wt$(C_1 \backslash C_2)$. The
resulting code is said pure if, in the above construction, $d_{x}
= d(C_{2}^{\perp})$ and $d_{z} = d(C_{1})$.
\end{lemma}

Since the Euclidean dual of a code $C$ and its Hermitian dual are
isomorphic under Galois conjugation that preserves Hamming metric,
a similar result can be derived if one considers in
Lemma~\ref{CSS1} the Hermitian inner product instead of
considering the Euclidean inner product and we shall call the
mentioned construction by CSS-type construction. Recently, the CSS
construction was extended to include additive codes \cite[Theorem
4.5]{Ezerman:2010}.

The following result shown in \cite{Ketkar:2006} will be utilized
in this paper:

 \begin{theorem}\cite[Theorem 13]{Ketkar:2006}
An $((n, K, d))_{q}$ stabilizer code exists if and only if there
exists an additive code $C\leq {\mathbb F}_{q}^{2n}$ of size
$|C|=q^{n}/K$ such that $C\leq C^{{\perp}_{s}}$ and
swt$(C^{{\perp}_{s}}\backslash C) = d$ if $K > 1$ (and
swt$(C^{{\perp}_{s}})=d$ if $K=1$).
\end{theorem}

\section{Construction Methods}\label{sec4}

This section is devoted to construct new AQECC from old ones. More
precisely, we show how to obtain new codes by extending,
puncturing, expanding, applying the direct sum and, finally, by
using the $({ \bf u}| { \bf u}+{ \bf v})$ construction. In other
words, we extend to AQECC all those methods valid to QECC.

\subsection{Code Expansion}\label{sub4.1}

Let us recall the concept of dual basis \cite{Lidl:1997}. Given a
basis $\beta=\{ b_1, b_2, \ldots, b_{m} \}$ of ${\mathbb
F}_{q^{m}}$ over ${\mathbb F}_{q}$, a \emph{dual basis} of $\beta$
is given by ${\beta}^{\perp}=\{ {b_1}^{\ast}, {b_2}^{\ast},\ldots,
{b_{m}}^{\ast} \}$, with tr$_{q^{m}/q}(b_{i}{b_{j}}^{\ast})
=\delta_{ij}$, for all $i,j\in \{1,\ldots, m \}$. A self-dual
basis $\beta$ is a basis satisfying $\beta={\beta}^{\perp}$. If
$C$ is an ${[n, k, d_1]}_{q^{m}}$ code and $\beta=\{b_1, b_2,
\ldots, b_{m}\}$ is a basis of ${\mathbb F}_{q^{m}}$ over
${\mathbb F}_{q}$, then the $q$-ary expansion $\beta(C)$ of $C$
with respect to $\beta$ is an ${[mn, mk, d_{2} \geq d_1]}_{q}$
code given by $\beta(C):= \{ {(c_{ij})}_{i,j}\in {{\mathbb
F}_{q}}^{mn} \mid\textbf{c}= {(\sum_{j}^{} c_{ij} b_{j} {)}_{i}}
\in C \}$.

\begin{lemma}\label{dualcode}\cite{Grassl:1999,Ashikhmin:2000,LaGuardia:2012}
Let $C={[n, k, d]}_{q^{m}}$ be a linear code over ${\mathbb
F}_{q^{m}}$, where $q$ is a prime power. Let ${C}^{\perp}$ be the
dual of the code $C$. Then the dual code of the $q$-ary expansion
$\beta(C)$ of code $C$ with respect to the basis $\beta$ is the
$q$-ary expansion ${{\beta}^{\perp}}({C}^{\perp})$ of the dual
code ${C}^{\perp}$ with respect to ${\beta}^{\perp}$.
\end{lemma}

Theorem~\ref{MAINI} presents a method to construct AQECC by
expanding linear codes:

\begin{theorem}\label{MAINI}
Let $q$ be a prime power. Assume that there exists an AQECC with
parameters ${[[n, k, d_{z}/ d_{x}]]}_{q^{m}}$, derived from linear
codes $C_1={[n, k_{1}, d_{1} ]}_{q^{m}}$ and $C_{2}={[n, k_{2},
d_{2}]}_{q^{m}}$, respectively. Then there exists an AQECC with
parameters ${[[mn, mk, d_{z}^{*}/ d_{x}^{*}]]}_{q}$, where
$k=k_{1} - k_{2}$, $d_{z}^{*}\geq d_1$ and $d_{x}^{*}\geq
d_{2}^{\perp}$, where $d_{2}^{\perp}$ denotes the minimum distance
of the dual code $C_{2}^{\perp}$.
\end{theorem}
\begin{proof}
The proof presented here utilizes the same idea and generalizes
the proof of \cite[Theorem 1]{LaGuardia:2012} to all linear codes.
We begin by observing that ${[\beta(C)]}^{\perp}= \beta^{\perp}
(C^{\perp})$. Let $C_1={[n, k_{1}, d_{1} ]}_{q^{m}}$ and $C_2={[n,
k_{2}, d_{2}]}_{q^{m}}$ be two codes such that $C_{2}\subset
C_{1}$. Let $\beta$ be any basis of ${\mathbb F}_{q^{m}}$ over
${\mathbb F}_{q}$ and ${\beta}^{\perp}$ its dual basis. Consider
the expansions $\beta(C_{1})$ of $C_1$ and $\beta(C_{2})$ of
$C_{2}$ with respect to $\beta$. Then the inclusion $\beta(C_{2})
\subset \beta(C_{1})$ holds. The codes $\beta(C_1),$ $\beta(C_2)$
and ${[\beta({C_2})]}^{\perp}$ are linear. Further,
$\beta(C_{1})={[mn, mk_{1}, D_{1}\geq d_1 ]}_{q}$ and
$\beta(C_{2})={[mn, m k_{2}, D_{2}\geq d_{2}]}_{q}$, respectively.
Since $C_{2}^{\perp}$ has minimum distance $d_{2}^{\perp}$, then
$\beta^{\perp} (C_2^{\perp})$ has minimum distance greater than or
equal to $d_{2}^{\perp}$ ($\beta^{\perp}$ is a basis of ${\mathbb
F}_{q^{m}}$ over ${\mathbb F}_{q}$). From Lemma~\ref{dualcode} the
equality ${[\beta({C_2})]}^{\perp}= \beta^{\perp} (C_2^{\perp})$
holds, hence ${[\beta(C_{2})]}^{\perp}$ also has minimum distance
greater than or equal to $d_{2}^{\perp}$. Applying the CSS
construction to $\beta(C_{1})$, $\beta(C_{2})$ and
${[\beta({C_{2}})]}^{\perp}$, one obtains an ${[[mn,
m(k_{1}-k_{2}), d_{z}^{*}/ d_{x}^{*}]]}_{q}$ asymmetric quantum
code, where $d_{z}^{*}\geq d_{1}$ and $d_{x}^{*}\geq
d_{2}^{\perp}$.
\end{proof}

More generally one has the following result:

\begin{theorem}\label{GenExpa}
Let $q=p^{t}$ be a prime power. If there exists an $((n, K,
d_{z}/d_{x}))_{q^{m}}$ stabilizer code then there exists an $((nm,
K, d_{z}^{*}/d_{x}^{*}))_{q}$ stabilizer code, where
$d_{z}^{*}\geq d_{z}$ and $d_{x}^{*}\geq d_{x}$.
\end{theorem}
\begin{proof}
If $a$ is an element of ${ \mathbb F}_{q^{m}}$, we can expand $a$
with respect to a given basis $B=\{ {\beta}_1 , \ldots ,
{\beta}_{m} \}$ of ${ \mathbb F}_{q^{m}}$ over ${\mathbb F}_{q}$
and put the coordinates of $a$ in the vector form $c_{B}(a)=(a_1 ,
\ldots , a_{m}) \in { \mathbb F}_{q}^{m}$. Consider the
non-degenerate symmetric form tr$_{q^{m}/q}(ab)$ on the vector
space ${ \mathbb F}_{q^{m}}$ (over ${\mathbb F}_{q}$). Assume that
${\varphi}_{B}$ is the ${\mathbb F}_{p}$-vector space isomorphism
from ${\mathbb F}_{q^{m}}^{2n}$ to ${\mathbb F}_{q}^{2nm}$ given
(in the proof of \cite[Lemma 76]{Ketkar:2006}) by
${\varphi}_{B}(({ \bf u}|{ \bf v}))=((c_{B}(u_1), \ldots ,
c_{B}(u_{n}))|( M c_{B}(v_1), \ldots , M c_{B}(v_{n})))$, where ${
\bf u}, \ { \bf v} \in {\mathbb F}_{q^{m}}^{n}$ are given by ${
\bf u}= (u_1, \ldots , u_n)$ and ${ \bf v}= (v_1, \ldots , v_n)$,
$M=($tr$_{q^{m}/q}({{\beta}_{i}}{{\beta}_{j}}))_{1\leq i, j\leq
m}$ denotes the Gram matrix and
tr$_{q^{m}/q}(ab)={c_{B}(a)}^{t}Mc_{B}(b)$ for all $a, b \in {
\mathbb F}_{q^{m}}$. Note that the inner product considered here
is the usual (Euclidean) inner product of ${\mathbb F}_{q}$.

Assume that an $((n, K, d_{z}/d_{x}))_{q^{m}}$ stabilizer code
exists. From \cite[Theorem 13]{Ketkar:2006}, there exists an
additive code $C\leq {\mathbb F}_{q^{m}}^{2n}$ of size
$|C|=q^{mn}/K$ such that $C\leq C^{{\perp}_{s}}$,
wt$_{X}(C^{{\perp}_{s}}\backslash C) = d_{x}$ if $K > 1$ (and
wt$_{X}(C^{{\perp}_{s}})=d_{x}$ if $K=1$) and
wt$_{Z}(C^{{\perp}_{s}}\backslash C) = d_{z}$ if $K > 1$ (and
wt$_{Z}(C^{{\perp}_{s}})=d_{z}$ if $K=1$). We know that
${\varphi}_{B}$ preserves trace-symplectic orthogonality, i. e.,
the code ${\varphi}_{B}(C)$ satisfies ${\varphi}_{B}(C)\leq
[{\varphi}_{B}(C)]^{{\perp}_{s}}$. If $({ \bf u}|{ \bf v}) \in
{\mathbb F}_{q^{m}}^{2n}$ and $u_{i}\neq 0$ (resp. $v_{j}\neq 0$)
for some $i \in \{1, \ldots , n\}$ (resp. $j \in \{1, \ldots ,
n\}$), then at least one coordinate of the corresponding vector
$c_{B}(u_{i})$ (resp. $M c_{B}(v_{j})$) is nonzero. Thus
wt$_{X}([{\varphi}_{B}(C)]^{{\perp}_{s}}\backslash
{\varphi}_{B}(C)) \geq d_{x}$ if $K > 1$ (and
wt$_{X}([{\varphi}_{B}(C)]^{{\perp}_{s}})\geq d_{x}$ if $K=1$) and
wt$_{Z}([{\varphi}_{B}(C)]^{{\perp}_{s}}\backslash
{\varphi}_{B}(C)) \geq d_{z}$ if $K
> 1$ (and wt$_{Z}([{\varphi}_{B}(C)]^{{\perp}_{s}}) \geq d_{z}$ if $K=1$).
Because the alphabet considered now is ${ \mathbb F}_{q}$, then
there exists an $((nm, K, d_{z}^{*}/d_{x}^{*}))_{q}$ stabilizer
code, where $d_{z}^{*}\geq d_{z}$ and $d_{x}^{*}\geq d_{x}$.
\end{proof}

\subsection{Direct Sum Codes}\label{sub4.2}

Let us recall the direct sum of codes. Assume that $C_1={[n_{1},
k_{1}, d_{1} ]}_{q}$ and $C_{2}={[n_{2}, k_{2}, d_{2}]}_{q}$ are
two linear codes. Then the direct sum code $C_{1}\oplus C_{2}$ is
the linear code given by $C_{1}\oplus C_{2}=\{ ({{\bf c}}_{1},
{{\bf c}}_{2})| {{\bf c}}_{1} \in C_{1}, {{\bf c}}_{2}\in C_{2}
\}$ and has parameters $[n_1+n_{2}, k_{1}+k_{2}, \min$
$\{d_{1},d_{2}\} {]}_{q}$.

\begin{theorem}\label{MAINII}
Let $q$ be a prime power. Assume there exists an AQECC with parameters
${[[n, k, d_{z}/ d_{x}]]}_{q}$ derived from linear codes
$C_{1}=[n, k_{1},$ $d_{1}{]}_{q}$ and $C_{2}={[n, k_{2}, d_{2}]}_{q}$ with $C_{2}
\subset C_1$. Suppose also there exists an $[[n^{*}, k^{*},$ $d_{z}^{*}/
d_{x}^{*}]{]}_{q}$ AQECC derived from classical linear codes
$C_{3}={[n^{*}, k_{3}, d_3]}_{q}$ and $C_{4}={[n, k_{4}, d_{4}]}_{q}$
with $C_{4} \subset C_3$. Then there exists an
${[[n+n^{*}, k+ k^{*}, d_{z}^{\diamond}/
d_{x}^{\diamond }]]}_{q}=[[n+n^{*}, (k_{1}+$ $k_3 )-(k_{2} + k_{4}), d_{z}^{\diamond}/
d_{x}^{\diamond }]{]}_{q}$ AQECC, where $d_{z}^{\diamond}\geq
\min \{ d_1 , d_{3} \}$, $d_{x}^{\diamond}\geq \min \{ d_{2}^{\perp },
d_{4}^{\perp }\}$ and $d_{2}^{\perp}$, $d_{4}^{\perp }$ are the
minimum distances of the dual codes $C_{2}^{\perp }$ and $C_{4}^{\perp }$,
respectively.
\end{theorem}

\begin{proof}
Consider the direct sum codes $C_{1}\oplus C_{3}={[n+n^{*},
k_{1}+k_{3}, \min \{d_{1},d_{3}\} ]}_{q}$ and $C_{2}\oplus
C_{4}={[n+n^{*}, k_{2}+k_{4}, \min \{d_{2}, d_{4}\} ]}_{q}$. Since
the inclusions $C_{2}\subset C_{1}$ and $C_{4}\subset C_3$ hold it
follows that the inclusion $C_{2} \oplus C_{4} \subset C_{1}
\oplus C_3$ also holds. We know that a parity check matrix of the
code ${(C_{2} \oplus C_{4})}^{\perp }$ is given by $G_{2}\oplus
G_{4}= \left[
\begin{array}{cc}
G_{2} & 0\\
0 & G_{4}\\
\end{array}
\right].$ Thus the minimum distance of ${(C_{2} \oplus
C_{4})}^{\perp }$ is equal to $\min \{ d_{2}^{\perp },
d_{4}^{\perp } \}$. Therefore, applying the CSS construction to
the codes $C_{1} \oplus C_3$, $C_{2} \oplus C_{4}$ and ${(C_{2}
\oplus C_{4})}^{\perp }$ one obtains an $[[n+n^{*}, (k_{1}+k_3 )$
$-(k_{2} + k_{4}), d_{z}^{\diamond}/ d_{x}^{\diamond }]{]}_{q}$
AQECC, where $d_{z}^{\diamond}\geq\min \{ d_1 , d_{3} \}$ and
$d_{x}^{\diamond}\geq \min \{ d_{2}^{\perp }, d_{4}^{\perp }\}$.
\end{proof}

The previous result also holds in a more general setting:

\begin{theorem}
Assume that there exist two stabilizer codes with parameters
$((n_{1}, K_1, d_{z}^{(1)}/d_{x}^{(1)}))_{q}$ and $((n_{2}, K_{2},
d_{z}^{(2)}/d_{x}^{(2)}))_{q}$. Then there exists an
$((n_{1}+n_{2}, K_{1}K_{2}, d_{z}^{*}/d_{x}^{*}))_{q}$, where
$d_{z}^{*} = \min \{d_{z}^{(1)}, d_{z}^{(2)} \}$ and $d_{x}^{*} =
\min \{d_{x}^{(1)}, d_{x}^{(2)} \}$.
\end{theorem}

\begin{proof}
The proof follows the same line of \cite[Lemma 73]{Ketkar:2006}.
We only show the result in the case of $X$-weight (the proof for
$Z$-weight is similar). Note that if $((n_{1}, K_1,
d_{z}^{(1)}/d_{x}^{(1)}))_{q}$ and $((n_{2}, K_{2},
d_{z}^{(2)}/d_{x}^{(2)}))_{q}$ are stabilizer codes with
orthogonal projectors $P_1$ and $P_{2}$ respectively, and
stabilizer $S_{1}$ and $S_{2}$ respectively, then $P_{1}\otimes
P_{2}$ is an orthogonal projector onto a $K_{1}K_{2}$-dimensional
subspace $Q^{\oplus}$ of ${\mathbb C}^{q^{(n_{1}+n_{2})}}$, and
the stabilizer of $Q^{\oplus}$ is given by
$S^{\oplus}=\{E_{1}\otimes E_{2} | E_{1} \in S_{1}, E_{2} \in
S_{2}\}$.  Assume that $F_1 \otimes F_{2}\in {{ \bf
G}}_{n_1}\otimes {{ \bf G}}_{n_{2}}$ is not detectable; hence $F_1
\in C_{{{ \bf G}}_{n_1}}(S_1)$ and $F_{2} \in C_{{{ \bf
G}}_{n_{2}}}(S_{2})$. Moreover, either $F_1 \notin S_{1}Z({{ \bf
G}}_{n_1})$ or $F_{2} \notin S_{2}Z({{ \bf G}}_{n_{2}})$,
otherwise $F_1 \otimes F_{2}$ would be detectable. Thus, from
\cite[Lemma 11]{Ketkar:2006}, either $F_1$ or $F_2$ is not
detectable, so wt$_{X}(F_1 \otimes F_{2})$ is at least $\min
\{d_{x}^{(1)}, d_{x}^{(2)} \}$, and the result follows.
\end{proof}

\subsection{Puncturing Codes}\label{sub4.3}

The technique of puncturing codes is well-known in the literature
as in the classical \cite{Macwilliams:1977,Huffman:2003} as well
as in the quantum case
\cite{Calderbank:1998,Rains:1999,Ketkar:2006}. In this section we
show how to construct AQECC by puncturing classical codes.

Let $C$ be an ${[n, k, d]}_{q}$ code. Then we denote by
$C^{P_{i}}$ the punctured code in the coordinate $i$. Recall that
the dual of a punctured code is a shortened code. Now we are ready
to show the main result of this subsection:

\begin{theorem}\label{MAINIII}
Assume that there exists an ${[[n, k, d_{z}/d_{x}]]}_{q}$
stabilizer code derived from two linear codes $C_{1}={[n, k_{1},
d_{1}]}_{q}$ and $C_{2}={[n, k_{2}, d_{2}]}_{q}$ with $C_{2}
\subset C_1$, $n \geq 2$, $k=k_1 - k_{2}$, $d_{z}\geq d_1$ and
$d_{x}\geq d_{2}^{\perp}$, where $d_{2}^{\perp}$ is the minimum
distance of the dual code $C_{2}^{\perp}$. Suppose also that
$d_{1}\geq 2$, $d_{2}^{\perp}\geq 2$ and $C_{2}^{\perp}$ contains
at least a nonzero codeword with $i$th coordinate zero. Then the
following hold:
\begin{enumerate}

\item [ (i)] If $C_1$ has a minimum weight codeword
with a nonzero $i$th coordinate then there exists an ${[[n-1, k,
d_{z}^{P_{i}}/d_{x}^{P_{i}}]]}_{q}$ AQECC, where $k=k_1 - k_{2}$, $d_{z}^{P_{i}}
\geq d_1 -1$ and $d_{x}^{P_{i}}\geq d_{2}^{\perp}$;

 \item [ (ii)] If $C_1$ has no minimum weight codeword with a nonzero $i$th
coordinate, then there exists an ${[[n-1, k,
d_{z}^{P_{i}}/d_{x}^{P_{i}}]]}_{q}$ AQECC, where $k=k_1 - k_{2}$, $d_{z}^{P_{i}}
\geq d_1$ and $d_{x}^{P_{i}}\geq d_{2}^{\perp}\geq 2$.
\end{enumerate}
\end{theorem}
\begin{proof}
We only prove item (ii) since the proof of (i) is similar to this
one. Consider the punctured codes $C_{1}^{P_{i}}$ and
$C_{2}^{P_{i}}$. Since the inclusion $C_{2}\subset C_1$ holds it
follows that $C_{2}^{P_{i}} \subset C_{1}^{P_{i}}$. Since from
hypothesis one has $d_{1}> 1$ then it follows that $d_{2} > 1$
because $C_{2} \subset C_1$; again from the hypothesis $C_1$ has
no minimum weight codeword with a nonzero $i$th coordinate. Thus,
by Theorem~\cite[Theorem 1.5.1]{Huffman:2003}, the punctured codes
$C_{1}^{P_{i}}$ and $C_{2}^{P_{i}}$ have parameters ${[n-1, k_1,
d_1]}_{q}$ and ${[n-1, k_{2}, d_{2}^{i}]}_{q}$, respectively,
where $d_{2}^{i}=d_{2}$ or $d_{2}^{i}=d_{2}-1$.

We need to compute the minimum distance of the code
${[C_{2}^{P_{i}}]}^{\perp}$ in order to apply the CSS
construction. To do this consider the code
${[C_{2}^{P_{i}}]}^{\perp}$. Since $C_{2}^{\perp}$ contains at
least a nonzero codeword whose $i$th coordinate is equal to zero
then $C_{2}^{\perp}$ has a subcode $C_{2}^{\perp}(\{i\}) \neq \{
{\bf 0 } \}$ and, consequently, the minimum distance
$d_{({C_{2}^{\perp})}_{i}}$ of $C_{2}^{\perp}(\{i\})$ satisfies
$d_{({C_{2}^{\perp})}_{i}}$ $\geq d_{2}^{\perp}$, where
$d_{2}^{\perp} > 1$. Since $d_{({C_{2}^{\perp})}_{i}}>1$ and
because (from definition) the code $C_{2}^{\perp}(\{i\})$ has no
minimum weight codeword with a nonzero $i$th coordinate, applying
again Theorem~\cite[Theorem 1.5.1]{Huffman:2003}, it implies that
the shortened code ${[C_{2}^{\perp}]}_{S_{i}}$ has minimum
distance equals $d_{({C_{2}^{\perp})}_{i}}$. From \cite[Theorem
1.5.7]{Huffman:2003} we know that ${[C_{2}^{P_{i}}]}^{\perp}=
{[C_{2}^{\perp}]}_{S_{i}}$, so the code
${[C_{2}^{P_{i}}]}^{\perp}$ has minimum distance
$d_{({C_{2}^{\perp})}_{i}}$, where $d_{({C_{2}^{\perp})}_{i}} \geq
d_{2}^{\perp}$. Therefore, applying the CSS construction to the
codes $C_{1}^{P_{i}}$, $C_{2}^{P_{i}}$ and
${[C_{2}^{P_{i}}]}^{\perp}$, one can derive an ${[[n-1, k,
d_{z}^{P_{i}} /d_{x}^{P_{i}}]]}_{q}$ AQECC, where $k=k_1 - k_{2}$,
$d_{z}^{P_{i}}\geq d_1$ and $d_{x}^{P_{i}} \geq
d_{({C_{2}^{\perp})}_{i}}\geq d_{2}^{\perp}\geq 2$.
\end{proof}

Following the lines adopted in \cite{Ketkar:2006} we can show a
more general result:

\begin{theorem}\label{MAINNEW}
Assume that a pure ${[[n, k, d_{z}/d_{x}]]}_{q}$ stabilizer code
exists, with $n \geq 2$ and $d_{x}, d_{z} \geq 2$. Then there
exists a pure ${[[n-1, k, d_{z}^{*}/d_{x}^{*}]]}_{q}$ stabilizer
code, where $d_{z}^{*}\geq d_{z}-1$ and $d_{x}^{*} \geq d_{x}-1$.
\end{theorem}
\begin{proof}
Assume that a pure ${[[n, k, d_{z}/d_{x}]]}_{q}$ stabilizer code
exists, with the corresponding minimum distance $d$. From
\cite[Corollary 72]{Ketkar:2006}, there exists a pure ${[[n-1, k,
d^{*}\geq d-1]]}_{q}$ stabilizer code derived from an additive
self-orthogonal (with respect to the trace-alternating form) code
$D^{{\perp}_{a}}\leq {\mathbb F}_{q^{2}}^{n-1}$ with
wt$(D^{{\perp}_{a}})\geq d-1$. Consider the vectors ${ \bf v}, {
\bf w} \in {\mathbb F}_{q}^{2(n-1)}$ and let $( \beta,
{\beta}^{q})$ be a normal basis of ${\mathbb F}_{q}^{2}$ over
${\mathbb F}_{q}$. We know that the bijective map $\phi (({ \bf
v}|{ \bf w}))= \beta{ \bf v} + {\beta}^{q}{ \bf w}$ from ${\mathbb
F}_{q}^{2(n-1)}$ onto ${\mathbb F}_{q^{2}}^{n-1}$ is an isometry
(symplectic/Hamming weights, resp.) (see also \cite[Lemma
14]{Ketkar:2006}). Considering the inverse map ${\phi}^{-1}$ and
the corresponding additive code ${\phi}^{-1}(D^{{\perp}_{a}}) \leq
{\mathbb F}_{q}^{2(n-1)}$, it follows that
${\phi}^{-1}(D^{{\perp}_{a}})$ has minimum $X$-weight $d_{x}^{*}$
at least $d_{x}^{*} \geq d_{x}-1$ and the minimum $Z$-weight
$d_{z}^{*}$ at least $d_{z}^{*} \geq d_{z}-1$, and the proof is
complete.
\end{proof}

\begin{remark}
Note that the procedure adopted in Theorems~\ref{MAINIII} and
\ref{MAINNEW} can be generalized by puncturing codes on two or
more coordinates.
\end{remark}

\subsection{Code Extension}\label{sub4.4}

The technique of (classical) code extension
\cite{Macwilliams:1977,Huffman:2003} was derived also in the
quantum case \cite{Calderbank:1998,Ketkar:2006}. Here we extend to
AQECC the referred technique.

Let $C$ be an ${[n, k, d]}_{q}$ linear code over ${\mathbb
F}_{q}$. The extended code $C^{e}$ is the linear code given by
$C^{e} = \{ (x_1,\ldots, x_{n}, x_{n+1}) \in {\mathbb F}_{q}^{n+1}
|$ $(x_1,\ldots, x_{n})\in C, x_1 +\cdots +x_{n}+ x_{n+1}=0 \}$.
The code $C^{e}$ is linear and has parameters ${[n+1, k,
d^{e}]}_{q}$, where $d^{e}=d$ or $d^{e}=d+1$. Recall that a vector
${ \bf v}=(v_1, \ldots, v_{n}) \in F_{q}^{n}$ is called
\emph{even-like} if it satisfies the equality
$\displaystyle\sum_{i=1}^{n} v_{i}=0$, and \emph{odd-like}
otherwise. For an ${[n, k, d]}_{q}$ code $C$ the minimum weight of
the even-like codewords of $C$ are called \emph{minimum even-like
weight} and denoted by $d_{even}$ (or ${(d)}_{even}$). Similarly,
the minimum weight of the odd-like codewords of $C$ are called
\emph{minimum odd-like weight} and denoted by $d_{odd}$ (or
${(d)}_{odd}$).

Let us now prove the main result of this subsection.

\begin{theorem}\label{MAINIV}
Assume that there exists an ${[[n, k, d_{z}/ d_{x}]]}_{q}$ AQECC
derived from codes $C_1 = {[n, k_1, d_1]}_{q}$ and $C_{2}={[n,
k_{2}, d_{2}]}_{q}$, where $C_{2}\subset C_1$. Then the following
hold:

\begin{itemize}
\item[ (a)] If ${(d_{1})}_{even}\leq {(d_{1})}_{odd}$, then there
exists an ${[[n+1, k, d_{z}^{e}/ d_{x}^{e}]]}_{q}$ AQECC, where
$d_{z}^{e} \geq d_1$ and $d_{x}^{e}\geq {(d_{2}^{e})}^{\perp}$,
where ${(d_{2}^{e})}^{\perp}$ is the minimum distance of the dual
${(C_{2}^{e})}^{\perp }$ of the extended code $C_{2}^{e}$;

\item[ (b)] If ${(d_{1})}_{odd} < {(d_{1})}_{even}$, then there
exists an ${[[n+1, k, d_{z}^{e}/ d_{x}^{e}]]}_{q}$ AQECC, where
$d_{z}^{e}\geq d_1 +1$ and $d_{x}^{e}\geq {(d_{2}^{e})}^{\perp}$.
\end{itemize}
\end{theorem}

\begin{proof}
We only show item  (b), since (a) is similar. It is easy to see
that the inclusion $C_{2}^{e} \subset C_{1}^{e}$ holds. The
parameters of the extended codes $C_{1}^{e}$ and $C_{2}^{e}$ are
${[n+1, k_1, d_{1}^{e}]}_{q}$ and ${[n+1, k_{2}, d_{2}^{e}]}_{q}$,
respectively, where $d_{1}^{e}=d_1$ or $d_{1}^{e}=d_1 + 1$. Since
${(d_{1})}_{odd} < {(d_{1})}_{even}$, it follows from the remark
shown in \cite[pg. 15]{Huffman:2003} that $d_{1}^{e}=d_1 + 1$.
From hypothesis we know that $k=k_1 - k_{2}$, so the corresponding
CSS code also has dimension $k$. Applying the CSS construction to
the codes $C_{1}^{e}$, $C_{2}^{e}$ and ${(C_{2}^{e})}^{\perp }$,
one obtains an AQECC with parameters ${[[n+1, k, d_{z}^{e}/
d_{x}^{e}]]}_{q}$, where $d_{z}^{e}\geq d_1 +1$ and $d_{x}^{e}\geq
{(d_{2}^{e})}^{\perp}$.
\end{proof}

\subsection{The $({ \bf u}| { \bf u}+{ \bf v})$ Construction}\label{sub4.5}

The $({ \bf u}| { \bf u}+{ \bf v})$ construction
\cite{Macwilliams:1977,Huffman:2003} is an interesting method for
constructing new (classical) linear codes. Our intention is to
apply this technique in order to generate a similar construction
method for asymmetric quantum codes.

Let $C_1$ and $C_{2}$ be two linear codes of same length both over
${\mathbb F}_{q}$ with parameters ${[n, k_1, d_1]}_{q}$ and ${[n,
k_{2}, d_{2}]}_{q}$, respectively. Then by applying the $({ \bf
u}| { \bf u}+{ \bf v})$ construction one can generate a new code
$C = \{ ({ \bf u}, { \bf u}+{ \bf v}) | { \bf u} \in C_1 , { \bf
v} \in C_{2} \}$ with parameters ${[2n, k_1 + k_{2}, \min \{
2d_{1}, d_{2} \}]}_{q}$. To simplify the notation, we denote the
code produced by applying the $({ \bf u}| { \bf u}+{ \bf v})$
construction to the codes $C_{1}$ and $C_{2}$ by $(C_{1}| C_{1} +
C_{2})$.

Theorem~\ref{MAINV} is the main result of this subsection:

\begin{theorem}\label{MAINV} Assume that there exist two
asymmetric stabilizer codes $[[n, k^{*}, d_{z}^{*}/$
$d_{x}^{*}]{]}_{q}$, derived from codes $C_1={[n, k_{1}, d_{1}
]}_{q}$ and $C_{2}={[n, k_{2}, d_{2} ]}_{q}$ with $C_{2}\subset
C_{1}$, and ${[[n, k^{\diamond }, d_{z}^{\diamond }/
d_{x}^{\diamond }]]}_{q}$, derived from codes $C_{3}={[n, k_3,
d_{3}]}_{q}$ and $C_{4}={[n, k_{4}, d_{4}]}_{q}$ with
$C_{4}\subset C_3$. Then there exists an ${[[2n,
k^{*}+k^{\diamond}, d_{z}/ d_{x}]]}_{q}$ AQECC, where $d_{z}\geq
\min \{ 2d_{1}, d_3 \}$, $d_{x}\geq \min \{ 2d_{4}^{\perp},
d_{2}^{\perp } \}$, with $d_{z}^{*}\geq d_1$, $d_{x}^{*}\geq
d_{2}^{\perp }$, $d_{z}^{\diamond }\geq d_3$ and $d_{x}^{\diamond
}\geq d_{4}^{\perp}$, where $d_{2}^{\perp }$ and $d_{4}^{\perp}$
are the minimum distances of the dual codes $C_{2}^{\perp}$ and
$C_{4}^{\perp}$, respectively.
\end{theorem}

\begin{proof}
Since the inclusions $C_{2}\subset C_{1}$ and $C_{4}\subset C_3$
hold it follows that the inclusion $(C_{2}| C_{2} + C_{4}) \subset
(C_{1}| C_{1} + C_{3})$ also holds. We know that the codes
$(C_{2}| C_{2} + C_{4})$ and $(C_{1}| C_{1} + C_{3})$ have
parameters ${[2n, k_{2} + k_{4}, \min \{ 2d_{2}, d_{4} \}]}_{q}$
and ${[2n, k_1 + k_3, \min \{ 2d_{1}, d_3 \}]}_{q}$, respectively.
Let us compute the minimum distance of the dual code ${[(C_{2}|
C_{2} + C_{4})]}^{\perp }$. We know that a generator matrix of
${[(C_{2}| C_{2} + C_{4})]}^{\perp }$ is the matrix $\left[
\begin{array}{cc}
H_{2} & 0\\
-H_{4} & H_{4}\\
\end{array}
\right],$ where $H_{2}$ and $H_{4}$ are the parity check matrices
of $C_{2}$ and $C_{4}$, respectively. The codewords of ${[(C_{2}|
C_{2} + C_{4})]}^{\perp }$ are of the form $\{ ({ \bf u}-{ \bf v},
{ \bf v}) | { \bf u} \in C_{2}^{\perp} , { \bf v} \in
C_{4}^{\perp} \}$. Consider the codeword ${ \bf w}=({ \bf u}-{ \bf
v}, { \bf v})$. If ${ \bf u}=0$ then ${ \bf w}=(-{ \bf v}, { \bf
v})$, so the minimum weight of ${[(C_{2}| C_{2} + C_{4})]}^{\perp
}$ is given by $2d_{4}^{\perp}$. On the other hand, if ${ \bf
u}\neq 0$ then wt$({ \bf w})=$ wt$({ \bf u}-{ \bf v})+$ wt$({ \bf
v})=$ d$({ \bf u}, { \bf v}) +$ d$({ \bf v}, { \bf 0})\geq $ d$({
\bf u}, { \bf 0})=$ wt$({ \bf u})$. Thus the minimum weight is
given by $d_{2}^{\perp}$ and, consequently, the minimum distance
of ${[(C_{2}| C_{2} + C_{4})]}^{\perp }$ is equal to $\min \{
2d_{4}^{\perp}, d_{2}^{\perp } \}$. Applying the CSS construction
to the codes $(C_{2}| C_{2} + C_{4})$, $(C_{1}| C_{1} + C_{3})$
and ${[(C_{2}| C_{2} + C_{4})]}^{\perp }$, one obtains an $[[2n,
(k_{1}+k_3)-(k_{2}+k_{4}), d_{z} / d_{x}]{]}_{q}={[[2n,
k^{*}+k^{\diamond}, d_{z}/ d_{x}]]}_{q}$ asymmetric stabilizer
code, where $d_{z}\geq \min \{ 2d_{1}, d_3 \}$ and $d_{x}\geq \min
\{ 2d_{4}^{\perp}, d_{2}^{\perp } \}$, as required.

As an alternative proof (suggested by the referee), we also can
write the codewords of ${[(C_{2}| C_{2} + C_{4})]}^{\perp }$ in
the form $\{ ({ \bf u}+{ \bf v}, -{ \bf v}) | { \bf u} \in
C_{2}^{\perp} , { \bf v} \in C_{4}^{\perp} \}$, and because the
Hamming weights of ${ \bf v}$ and $-{ \bf v}$ are the same, the
latter code is equivalent to $\{ ({ \bf u}+{ \bf v}, { \bf v}) | {
\bf u} \in C_{2}^{\perp} , { \bf v} \in C_{4}^{\perp} \}$, and the
result follows.
\end{proof}

\section{Code Constructions}\label{sec5}

In this section we utilize the construction methods developed in
Section~\ref{sec4} to obtain new families of AQECC. In order to
shorten the length of this paper we only apply the quantum code
expansion shown in Subsection~\ref{sub4.1} of Section~\ref{sec4},
although it is clear that all construction methods proposed in
Section~\ref{sec4} can also be applied. In
Subsections~\ref{subsec5.1}, \ref{subsec5.2}, \ref{subsec5.3},
\ref{subsec5.4} and \ref{subsec5.5} we construct AQECC derived
from generalized Reed-Muller (GRM), character codes, BCH,
quadratic residue (QR) and affine-invariant codes, respectively.
In Subsection~\ref{subsec5.6}, we construct a code table
containing the parameters of known AQECC as well the parameters of
the new codes.

\begin{remark}
It is important to observe that in all results presented in the
following, we expand the codes defined over ${\mathbb F}_{q}$
(where $q=p^{t}$, $t\geq 1$ and $p$ prime) with respect to the
prime field ${\mathbb F}_{p}$. However, the method also holds if
one expands such a codes over any subfield of the field ${\mathbb
F}_{q}$.
\end{remark}

\subsection{Construction I- Generalized Reed-Muller Codes}\label{subsec5.1}

The first family of AQECC derived from binary Reed-Muller (RM)
codes were constructed in \cite[Lemma 4.1]{Sarvepalli:2009}. In
this subsection we present a construction of AQECC derived from
generalized Reed-Muller (GRM)
\cite{Macwilliams:1977,Peterson:1972}.

The GRM code ${\cal{R}}_{q}(\alpha , m)$ over ${\mathbb F}_{q}$ of
order $\alpha$, $0\leq \alpha < q(m-1)$, has parameters ${[q^{m},
k(\alpha), d(\alpha)]}_{q}$, where
\begin{eqnarray}\label{RMpar1}
k(\alpha)=\displaystyle\sum_{i=0}^{m}
{(-1)}^{i}\left(
\begin{array}{c}
m\\
i\\
\end{array}
\right)
\left(
\begin{array}{c}
m+\alpha -iq\\
\alpha -iq\\
\end{array}
\right)
\end{eqnarray}
and
\begin{eqnarray}\label{RMpar2}
d(\alpha)=(t+1)q^{u},
\end{eqnarray}
where $m(q-1)-\alpha =(q-1)u+t$ and $0\leq t < q-1$. The dual of a GRM code
${\cal{R}}_{q}(\alpha , m)$ is also a GRM code given by
$[{\cal{R}}_{q}(\alpha , m)]^{\perp } = {\cal{R}}_{q}({\alpha}^{\perp } , m)$, where
${\alpha}^{\perp } = m(q-1)-1 -\alpha$.

We use the properties of the GRM codes in order to deriving new
asymmetric quantum codes:

\begin{theorem}\label{mainGRM}
Let $0\leq {\alpha}_1 \leq {\alpha}_{2}< m(q-1)$ and assume that
$q=p^{t}$ is a prime power, where $t\geq 1$. Then there exists an
$p$-ary asymmetric quantum GRM code with parameters ${[[tq^{m}, \
t[k({\alpha}_{2})-k({\alpha}_{1})], \ d_{z}/d_{x}]]}_{p}$, where
$d_{z}\geq d({\alpha}_{2})$, $d_{x}\geq d({\alpha}_{1}^{\perp})$,
$k({\alpha}_{2})$ and $k({\alpha}_{1})$ are given in
Eq.~(\ref{RMpar1}), $d({\alpha}_{2})$ is given in
Eq.~(\ref{RMpar2}) and $d({\alpha}_{1}^{\perp})=(a+1)q^{b}$, where
${\alpha}_{1}+1 =(q-1)b+a$ and $0\leq a\leq q-1$.
\end{theorem}

\begin{proof}
First, note that since the inequality ${\alpha}_1 \leq
{\alpha}_{2}$ holds then the inclusion ${\cal{R}}_{q}({\alpha}_{1}
, m) \subset {\cal{R}}_{q}({\alpha}_{2}, m)$ also holds. The codes
$\beta ({\cal{R}}_{q}({\alpha}_{1}, m))$ and $\beta
({\cal{R}}_{q}({\alpha}_{2}, m))$ have parameters ${[tq^{m},
tk({\alpha}_{1}), d({\alpha}_{1})]}_{p}$ and ${[tq^{m},
tk({\alpha}_{2}), d({\alpha}_{2})]}_{p}$, respectively, where
$k({\alpha}_{1})$ and $k({\alpha}_{2})$ are computed according to
Eq. (\ref{RMpar1}) and $d({\alpha}_{1})$, $d({\alpha}_{2})$ are
computed by applying Eq. (\ref{RMpar2}). We know that the
parameter ${\alpha}_{1}^{\perp }$ of the dual code
$[{\cal{R}}_{q}({\alpha}_{1}, m)]^{\perp }$ $={\cal{R}}_{q}
({\alpha}_{1}^{\perp}, m)$ equals ${\alpha}_{1}^{\perp
}=m(q-1)-1-{\alpha}_{1}$, so the minimum distance of
$[{\cal{R}}_{q}({\alpha}_{1} , m)]^{\perp }$ is equal to
$d({\alpha}_{1}^{\perp})=(a+1)q^{b}$, where ${\alpha}_{1}+1
=(q-1)b+a$ and $0\leq a\leq q-1$. Thus the code
${[\beta({\cal{R}}_{q}({\alpha}_{1} , m))]}^{\perp}$ has minimum
distance greater than or equal to $d({\alpha}_{1}^{\perp})$.
Applying Theorem~\ref{MAINI} one can get an ${[[tq^{m}, \
t[k({\alpha}_{2})-k({\alpha}_{1})], \ d_{z}/d_{x}]]}_{p}$
asymmetric stabilizer code, where $d_{z}\geq d({\alpha}_{2})$ and
$d_{x}\geq d({\alpha}_{1}^{\perp})$.
\end{proof}

\subsection{Construction II- Character Codes}\label{subsec5.2}

The class of (classical) character codes were introduced by Ding
\emph{et al.} \cite{Ding:2000}. Let us consider the commutative
group $G={\mathbb{Z}}_{2}^{m}$, $m\geq 1$ and a finite field
${\mathbb F}_{q}$ of odd characteristic. Recall that the code
$C_{q}(r, m)=C_{X}$, where $X\subset {\mathbb{Z}}_{2}^{m}$
consists of elements with Hamming weight greater than $r$ has
parameters $[2^{m},$ $s_{m}(r), 2^{m-r}{]}_{q}$ (see \cite[Theorem
6]{Ding:2000}), where $s_{m}(r)= \displaystyle\sum_{i=0}^{r}\left(
\begin{array}{c}
m\\
i\\
\end{array}
\right)$. The (Euclidean) dual code ${[C_{q}(r, m)]}^{\perp}$ of
$C_{q}(r, m)$ is equivalent to $C_{q}(m-r-1, m)$ (see
\cite[Theorem 8]{Ding:2000}) and consequently has parameters
$[2^{m}, s_{m}(m-r-1),$ $2^{r+1}{]}_{q}$.

Next we utilize the code expansion applied to character codes to
generate new AQECC, as established in the following theorem:

\begin{theorem}\label{lagchar}
If $0\leq r_1 < r_{2} \leq  m$ and $q=p^{t}$ is a power of an odd prime $p$,
where $t\geq 1$, then there exists an ${[[t2^{m}, t[k(r_{2}) - k(r_1)],
d_{z}/d_{x}]]}_{p}$ AQECC, where $k(r)=\displaystyle\sum_{i=0}^{r}\left(
\begin{array}{c}
m\\
i\\
\end{array}
\right)$
and $d_{z}\geq 2^{m-r_{2}}$ and $d_{x}\geq  2^{r_1 + 1}$.
\end{theorem}

\begin{proof}
It is easy to see that $C_{q}(r_{1}, m)\subset C_{q}(r_{2}, m)$.
The dual code ${[C_{q}(r_{1}, m)]}^{\perp}$ is equivalent to the
code $C_{q}(m-r_{1}-1, m)$. Applying Theorem~\ref{MAINI} one can
get an ${[[t2^{m}, t(k(r_{2}) - k(r_1)), d_{z}/d_{x}]]}_{p}$
AQECC, where $t$, $k(r_{1})$, $k(r_{2})$, $d_{x}$ and $d_{z}$ are
specified in the hypothesis.
\end{proof}

\subsection{Construction III - BCH Codes}\label{subsec5.3}

In this subsection we construct more families of asymmetric
stabilizer codes derived from Bose-Chaudhuri-Hocquenghem (BCH)
codes \cite{Macwilliams:1977}. The first families of AQECC derived
from BCH codes were constructed by Aly \cite[Theorem
8]{Salah:2008}. Recently, the parameters of these codes were
improved for certain families of BCH codes \cite{LaGuardia:2011}.

Recall that a cyclic code of length $n$ over ${\mathbb F}_{q}$ is
a BCH code with designed distance $\delta$ if, for some integer
$b\geq 0,$ one has $$g(x)= l.c.m. \{{M}^{(b)}(x), {M}^{(b+1)}(x),
\ldots, {M}^{(b+\delta-2)}(x)\},$$ i. e., $g(x)$ is the monic
polynomial of smallest degree over ${\mathbb F}_{q}$ having
${{\alpha}^{b}}, {{\alpha}^{b+1}},$ $\ldots, {{\alpha}^{b+
\delta-2}}$ as zeros. The next result shows how to construct more
AQECC by expanding (classical) BCH codes:

\begin{theorem}\label{ABCH1}
Suppose that $n=q^{m}-1$, where $q=p^{t}$ is a power of an odd prime $p$,
$t\geq 1$ and $m\geq 3$ are integers an integer (if $q=3$, $m\geq 4$).
Then there exist quantum codes with parameters

\begin{itemize}
\item  ${[[tn, t(n - m(4q-5) - 2), d_{z}\geq (2q+2) / d_{x}\geq 2q]]}_{p};$

\item $[[tn, t(n - m(4q-c-5) - 2), d_{z}\geq (2q+2) / d_{x}\geq$ $(2q-c)]{]}_{p},$
where $0\leq c\leq q-2$;

\item  ${[[tn, t(n -m(2c-l-4)-2), d_{z}\geq c / d_{x}\geq (c-l)]]}_{p},$ \\
where $2\leq c \leq q$ and $0\leq l\leq c-2$;

\item  ${[[tn, t(n -m(2c-l-6)-2), d_{z}\geq c / d_{x}\geq (c-l)]]}_{p},$ \\
where $q+2 < c \leq 2q$ and $0\leq l\leq c-q-3$;

\item  $[[tn, t(n -m(4q-l-5)-1), d_{z}\geq (2q+1) / d_{x}\geq$ $(2q-l)]{]}_{p},$
where $0\leq l\leq q-2$.
\end{itemize}
\end{theorem}
\begin{proof}
Consider the codes constructed in \cite[Theorems 4 and 5 and
Corollary 1]{LaGuardia:2011}. These codes are derived from two
distinct nested cyclic codes $C_{2}\subset C_1$. Thus, applying
Theorem~\ref{MAINI} the result holds.
\end{proof}

\begin{theorem}
Let $q=p^{t}$ be a power of a prime $p$, $t\geq 1$, $\gcd (q,
n)=1$ and ${{ord}_{n}}(q) = m$. Let $C_1$ and $C_2$ be two
narrow-sense BCH codes of length $q^{\lfloor m/2\rfloor}< n \leq
q^{m}-1$ over ${\mathbb F}^{q}$ with designed distances
${\delta}_{1}$ and ${\delta}_{2}$ in the range $2 \leq
{\delta}_{1}, {\delta}_{2} \leq {\delta}_{max}=\min \{ \lfloor
nq^{\lceil m/2\rceil}/(q^{m}-1)\rfloor , n \}$ and ${\delta}_{1}<
{\delta}_{2}^{\perp}\leq {\delta}_{2}< {\delta}_{1}^{\perp}$.
Assume also that $S_1 \cup\ldots \cup S_{{\delta}_{1}-1}\neq S_1
\cup\ldots \cup S_{{\delta}_{2}-1}$, where $S_{i}$ denotes a
cyclotomic coset. Then there exists an AQECC with parameters
$[[tn, t(n -m\lceil({\delta}_1 -1)(1-1/q) \rceil
-m\lceil({\delta}_{2}-1) (1-1/q)\rceil),$
$d_{z}^{*}/d_{x}^{*}]{]}_{p},$ where $d_{z}^{*}=wt(C_{2}\backslash
C_1^{\perp}) \geq {\delta}_{2}$ and $d_{x}^{*} = wt(C_1\backslash
C_2^{\perp})\geq {\delta}_1$.
\end{theorem}
\begin{proof}
It suffices to apply Theorem~\ref{MAINI} in those codes shown in
\cite[Theorem 8]{Salah:2008}.
\end{proof}

\begin{remark}
Note that one can obtain more families of AQECC by applying
Theorem~\ref{MAINI} in the existing families shown in
\cite{LaGuardia:2009}. Moreover, expanding generalized
Reed-Solomon (GRS) codes, one obtains \cite[Theorem
7.1]{LaGuardia:2012} as a particular case of Theorem~\ref{MAINI}.
\end{remark}

\subsection{Construction IV- Quadratic Residue Codes}\label{subsec5.4}

In this subsection we construct families of AQECC derived from
quadratic residue (QR) codes \cite{Macwilliams:1977,Huffman:2003}.
A family of quantum codes derived from classical QR codes was
constructed in \cite[Theorems 40 and 41]{Ketkar:2006}.

Let $p$ be an odd prime not dividing $q$, where $q$ is a prime
power that is a square modulo $p$. Let $Q$ be the set of nonzero
squares modulo $p$ and $C$ consisting of non-squares modulo $p$.
The quadratic residue codes $\mathcal{Q}$,
${\mathcal{Q}}^{\diamond }$, $\mathcal{C}$ and
${\mathcal{C}}^{\diamond}$ are cyclic codes with generator
polynomials $ q(x), ~~ (x - 1)q(x), ~~ c(x), ~~  (x - 1)c(x)$,
respectively, where $q(x) = \displaystyle \prod_{r\in Q}(x -
{\alpha}^{r}), \ \ \ c(x) = \displaystyle \prod_{s\in C}(x -
{\alpha}^{s})$ have coefficients from ${\mathbb F}_{q}$, and
$\alpha$ is a primitive $p$th root of unity belonging to some
extension field of ${\mathbb F}_{q}$. The codes $\mathcal{Q}$ and
$\mathcal{C}$ have the same parameters ${[p, (p+1)/2, d_1]}_{q}$,
where ${(d_{1})}^{2}\geq p$; similarly, the codes
${\mathcal{Q}}^{\diamond}$ and ${\mathcal{C}}^{\diamond}$ also
have the same parameters ${[ p, (p-1)/2, d_{2} ]}_{q}$, where
${(d_{2})}^{2}\geq p$.

Now we construct families of AQECC by expanding quadratic residue
codes:

\begin{theorem}\label{qrexp1} Let $p$ be a prime of the form
$p\equiv 1\mod 4$, and let $q=p_{*}^{t}$ ($t\geq 1$) be a power of
a prime that is not divisible by $p$. If $q$ is a quadratic
residue modulo $p$, then there exists an ${[[tp, t,
d_{z}/d_{x}]]}_{p_{*}}$ asymmetric quantum code, where $d_{z}$ and
$d_{x}$ satisfy $d_{z}\geq \sqrt{p}$ and $d_{x}\geq \sqrt{p}$.
\end{theorem}

\begin{proof}
Consider the codes $\mathcal{Q}$, ${\mathcal{Q}}^{\diamond }$ and
$\mathcal{C}$ given above. Since $p=4k+1$, then it is well known
that ${\mathcal{Q}}^{\diamond}={\mathcal{C}}^{\perp}$, so
${\mathcal{C}}^{\perp}\subset \mathcal{Q}$. The codes
$\mathcal{Q}$ and ${\mathcal{C}}^{\perp}$ have parameters,
respectively, given by ${[p, (p+1)/2, d_{1}]}_{q}$, with
${(d_1)}^{2}\geq p$ and ${[p, (p-1)/2, d_{2}]}_{q}$, where
${(d_{2})}^{2}\geq p$. Proceeding similarly as in the proof of
Theorem~\ref{MAINI} one can get an ${[[tp, t,
d_{z}/d_{x}]]}_{p_{*}}$ asymmetric quantum code, where $d_{z}$ and
$d_{x}$ satisfy $d_{z}\geq \sqrt{p}$ and $d_{x}\geq \sqrt{p}$.
\end{proof}

\begin{theorem}\label{qrexp2}
Let $p$ be a prime of the form $p\equiv 3 \mod 4$, and let
$q=p_{*}^{t}$ ($t\geq 1$) be a power of a prime that is not
divisible by $p$. If $q$ is a quadratic residue modulo $p$, then
there exists an ${[[tp, t, d_{z}/d_{x}]]}_{p_{*}}$ quantum code,
where $d_{z} \geq d$, $d_{x} \geq d$ and $d$ satisfies $d^{2}-d+1
\geq p$.
\end{theorem}

\begin{proof}
Since $p=4k-1$, the dual ${\mathcal{Q}}^{\perp}$ of $\mathcal{Q}$
equals ${\mathcal{Q}}^{\perp}= {\mathcal{Q}}^{\diamond}$, so
${\mathcal{Q}}^{\perp}\subset \mathcal{Q}$. The codes
$\mathcal{Q}$ and ${\mathcal{Q}}^{\perp}$ have parameters ${[p,
(p+1)/2, d]}_{q}$ and ${[p, (p-1)/2, d^{\diamond}\geq d]}_{q}$,
respectively, and the minimum distance is bounded by
$d^{2}-d+1\geq p$ (see for instance the proof of Theorem 40 in
\cite{Ketkar:2006}). Applying Theorem~\ref{GenExpa} one has an
${[[tp, t, d_{z}/d_{x}]]}_{p_{*}}$ code, where $d_{z}\geq d$,
$d_{x} \geq d$ and $d^{2}-d+1\geq p$.
\end{proof}

\begin{remark}
As observed by the referee, a refined statement can be made if one
considers the code ${\mathcal{Q}}^{\diamond }$ instead of
considering the code ${\mathcal{Q}}$, because
$d_{{\mathcal{Q}}^{\diamond }}=d_{{\mathcal{Q}}}+1$ (see
\cite[Chapter 16, Problem (2), p. 494]{Macwilliams:1977}).
\end{remark}

\subsection{Construction V- Affine-Invariant Codes}\label{subsec5.5}

We assume that the reader is familiar with the class of
(classical) affine-invariant codes. The structure and results on
this class of codes can be found in \cite{Huffman:2003}.

Quantum affine-invariant codes were investigated in the literature
\cite{Guenda:2009}:

\begin{lemma}\cite[Lemma 22]{Guenda:2009}\label{af-inv1}
Let $C^{e}$ be an extended maximal affine-invariant code ${[p^{m},
p^{m}-1-m/t, d]}_{p^{t}}$, then if $p > 3$ or $m > 2$ or $t\neq
1$, we have ${(C^{e})}^{\perp}\subset C^{e}$.
\end{lemma}

Applying Lemma~\ref{af-inv1} we can construct a family of AQECC
derived from affine-invariant codes:

\begin{theorem}\label{af-inv-lag}
Assume that $q=p^{t}$, $m$ is a positive integer and $n=p^{m}-1$.
If $p > 3$ or $m > 2$ or $t\neq 1$ then there exists an AQECC com
parameters ${[[tp^{m}, t(p^{m}-2 -2\frac{m}{t}), d_{z}/d_{x}
]]}_{p}$, where $d_{z}\geq d_{a}$, $d_{x}\geq d_{a}$, and $d_{a}$
is the minimum distance of an extended maximal affine-invariant
code.
\end{theorem}

\begin{proof}
Consider the dual containing extended maximal affine-invariant
code $C^{e}$ with parameters ${[p^{m}, p^{m}-1-m/t, d]}_{}$ given
in Lemma~\ref{af-inv1}, where $p > 3$ (or $m > 2$ or $t\neq 1$).
Applying Theorem~\ref{GenExpa} one obtains an ${[[tp^{m},
t(p^{m}-2 -2\frac{m}{t}), d_{z}/d_{x} ]]}_{p}$ AQECC, where
$d_{z}\geq d_{a}$, $d_{x}\geq d_{a}$, and $d_{a}$ is the minimum
distance of $C^{e}$.
\end{proof}

\subsection{Code Tables}\label{subsec5.6}

In this section we present Tables~\ref{table1} and \ref{table2}
containing families of AQECC available in the literature as well
as the new code families constructed in this paper. In the first
column we give the class and the parameters ${[[n, k, d_{z}/
d_{x}]]}_{q}$ of an AQECC; in the second column the parameter's
range, and in third column, the corresponding references.

\section{Summary}\label{sec6}

We have shown how to construct new families of asymmetric
stabilizer codes by applying the techniques of puncturing,
extending, expanding, direct sum and the $({ \bf u}| { \bf u}+{
\bf v})$ construction. As examples of application of quantum code
expansion, new AQECC derived from generalized Reed-Muller,
quadratic residue, BCH, character and affine-invariant codes have
been constructed.

\section*{Acknowledgment}
I am indebted to the anonymous referee for their valuable comments
and suggestions that improve significantly the quality of this
paper. Additionally, he/she pointed out that the construction
methods presented in the first version of this manuscript
(asymmetric quantum codes derived from linear codes) also hold for
a more general setting (asymmetric quantum codes derived from
additive codes). Based on this suggestion, I have added more
results concerning asymmetric codes derived from additive codes.
This work was partially supported by the Brazilian Agencies CAPES
and CNPq.

\clearpage

\begin{table}[bth!]
\begin{center}
\caption{Families of AQECC
\label{table1}}
\begin{tabular}{|c| c| c|}
\hline Code Family / ${[[n, k, d_{z}/ d_{x}]]}_{q}$ & Range of Parameters & Ref.\\
\hline BCH & &\\
\hline
\hline ${[[n, n -m\lceil({\delta}_1 -1)(1-1/q) \rceil -m\lceil({\delta}_{2}-1)
(1-1/q)\rceil, \ d_{z}^{*}/d_{x}^{*}]]}_{q}$ & $\gcd (q, n)=1$,
${{ord}_{n}}(q) = m$, & \cite{Salah:2008}\\
\hline & $q^{\lfloor m/2 \rfloor} < n \leq  q^{m}-1$, &\\
\hline & $2 \leq {\delta}_1 , {\delta}_{2}\leq {\delta }_{max}=\min\{\lfloor n
q^{\lceil m/2 \rceil}/(q^{m}-1)\rfloor, n \}$, &\\
\hline & ${\delta}_1 < {\delta}_2^{\perp} \leq {\delta}_{2} < {\delta}_1^{\perp}$, &\\
\hline & $d_{z}^{*}=wt(C_{2}\backslash C_1^{\perp}) \geq {\delta}_{2}$, &\\
\hline & $d_{x}^{*} = wt(C_1\backslash C_2^{\perp})\geq {\delta}_1$ &\\
\hline
\hline ${[[2^{m}-1, m({\delta}_{2}-{\delta}_{1})/2, d_{x}/ d_{z}]]}_{q}$ & $m\geq 2$, $2\leq {\delta}_{1}<
{\delta}_{2} < {\delta}_{max}=2^{\lceil m/2\rceil}-1$, & \cite{Sarvepalli:2009}\\
\hline & ${\delta}_{i}\equiv 1$ mod$2$ $d_{x}\geq {\delta}_{1}$, $d_{z}\geq {\delta}_{max} + 1$ &\\
\hline
\hline ${[[n, k, d_{z}/ d_{x}]]}_{q}$ & $n=q^{m}-1$, $m\geq 3$ (if $q=3$, $m\geq 4$): & \cite{LaGuardia:2011}\\
\hline ${[[n, n - m(4q-5) - 2, d_{z}\geq (2q+2) / d_{x}\geq 2q]]}_{q}$ & & \\
\hline ${[[n, n - m(4q-c-5) - 2, d_{z}\geq (2q+2) / d_{x}\geq (2q-c)]]}_{q}$ & $0\leq c\leq q-2$ &\\
\hline ${[[n, n -m(2c-l-4)-2, d_{z}\geq c / d_{x}\geq (c-l)]]}_{q}$ & $2\leq c \leq q$ and $0\leq l\leq c-2$ &\\
\hline ${[[n, n -m(2c-l-6)-2, d_{z}\geq c / d_{x}\geq (c-l)]]}_{q}$ & $q+2 < c \leq 2q$ and $0\leq l\leq c-q-3$ &\\
\hline ${[[n, n -m(4q-l-5)-1, d_{z}\geq (2q+1) / d_{x}\geq (2q-l)]]}_{q}$ & $0\leq l\leq q-2$ & \\
\hline
\hline Expanded BCH & &\\
\hline
\hline ${[[tn, t[n -m\lceil({\delta}_1 -1)(1-1/q) \rceil -m\lceil({\delta}_{2}-1)
(1-1/q)\rceil], \ d_{z}^{*}/d_{x}^{*}]]}_{q}$ & $\gcd (q, n)=1$,
${{ord}_{n}}(q) = m$, & \cite{Salah:2008}\\
\hline & $t\geq 1$, $q^{\lfloor m/2 \rfloor} < n \leq  q^{m}-1$, &\\
\hline & $2 \leq {\delta}_1 , {\delta}_{2}\leq {\delta }_{max}=\min\{\lfloor n
q^{\lceil m/2 \rceil}/(q^{m}-1)\rfloor, n \}$, &\\
\hline & ${\delta}_1 < {\delta}_2^{\perp} \leq {\delta}_{2} < {\delta}_1^{\perp}$, &\\
\hline & $d_{z}^{*}=wt(C_{2}\backslash C_1^{\perp}) \geq {\delta}_{2}$, &\\
\hline & $d_{x}^{*} = wt(C_1\backslash C_2^{\perp})\geq {\delta}_1$ &\\
\hline
\hline ${[[tn, tk, d_{z}/ d_{x}]]}_{q}$ & $n=q^{m}-1$, $q=p^{t}$, $p$ odd prime, $t\geq 1$, &\\
\hline & $m\geq 3$ (if $q=3$, $m\geq 4$): &\\
\hline ${[[tn, t(n - m(4q-5) - 2), d_{z}\geq (2q+2) / d_{x}\geq 2q]]}_{p}$ & &\\
\hline $[[tn, t(n - m(4q-c-5) - 2), d_{z}\geq (2q+2) / d_{x}\geq$ $(2q-c)]{]}_{p}$ & $0\leq c\leq q-2$ &\\
\hline ${[[tn, t(n -m(2c-l-4)-2), d_{z}\geq c / d_{x}\geq (c-l)]]}_{p}$ & $2\leq c \leq q$, $0\leq l\leq c-2$ &\\
\hline ${[[tn, t(n -m(2c-l-6)-2), d_{z}\geq c / d_{x}\geq (c-l)]]}_{p}$ & $q+2 < c \leq 2q$, $0\leq l\leq c-q-3$\\
\hline $[[tn, t(n -m(4q-l-5)-1), d_{z}\geq (2q+1) / d_{x}\geq$ $(2q-l)]{]}_{p}$ & $0\leq l\leq q-2$ &\\
\hline
\hline BCH-LDPC & &\\
\hline
\hline ${[[p^{ms}-1, k_{x} + k_{z}-p^{ms}+1, d_{z}/ d_{x}]]}_{p}$ & $\delta \leq {\delta}_{0}= p^{{\mu}s}-1$ &
\cite{Sarvepalli:2009}\\
\hline & $k_{x}=\dim$ BCH$(\delta)\subseteq {\mathbb F}_{p}^{n}$, &\\
\hline & $k_{z}=\dim C_{EG, c}^{(1)} (m, {\mu}, 0, s, p)$, &\\
\hline & $d_{x}\geq \delta$, $d_{z}\geq  A_{EG} (m, {\mu}, {\mu}-1, s, p)$ &\\
\hline ${[[2^{2s}-1, 2^{2s}-3^{s}-s(\delta - 1), \delta / 2^{s} + 1]]}_{2}$ & $\delta=2t + 1 \leq 2^{s}-1$ & \cite{Sarvepalli:2009}\\
\hline
\hline ${[[n, k_{x} + k_{z}-n, d_{z}/ d_{x}]]}_{p}$ & $n=(p^{(m+1)s}-1)/(p^{s} -1)$ &\cite{Sarvepalli:2009}\\
\hline & $\delta \leq {\delta}_{0}=(p^{(\mu+1)s}-1)/(p^{s} -1)$, $k_{x}=\dim$ BCH$_{p}(\delta, n)$, &\\
\hline & $k_{z}=\dim C_{PG}^{(1)} (m, {\mu}, 0, s, p)$, $d_{x}\geq \delta$, &\\
\hline & $d_{z}\geq A_{EG} (m, {\mu}, {\mu}-1, s, p)$ &\\
\hline
\hline ${[[n, n-3^{s}-3s\lceil (\delta -1)/2\rceil -1, \delta/ (2^{s}+2)]]}_{2}$ &
$n=2^{2s}+2^{s}+1$, $\delta \leq 2^{s/2} + 1$ &\cite{Sarvepalli:2009}\\
\hline
\hline LDPC-LDPC & &\\
\hline ${[[p^{ms}, k_{x} + k_{z}-p^{ms}, d_{z}/ d_{x}]]}_{p}$ & $p$ prime,
$q=p^{s}$, $s\geq 1$, $m\geq 2$, & \cite{Sarvepalli:2009}\\
\hline & $1<{\mu}_{z} < m$, $ m - {\mu}_{z} +1\leq {\mu}_{x} < m$, &\\
\hline & $k_{x}=\dim C_{EG}^{(1)} (m, {\mu}_{x}, 0, s, p)$, &\\
\hline & $k_{z}=\dim C_{EG}^{(1)} (m, {\mu}_{z}, 0, s, p)$, &\\
\hline & $d_{x}\geq  A_{EG} (m, {\mu}_{x}, {\mu}_{x}-1, s, p)+1$, &\\
\hline & $d_{z}\geq  A_{EG} (m, {\mu}_{z}, {\mu}_{z}-1, s, p)+1$ &\\
\hline
\hline concatenated RS & &\\
\hline ${[[2mq, mk - 1, (\geq 2(q-k+1))/2]]}_{4}$ & $n=4^{m}$, $1\leq k\leq q$ & \cite{Ezerman:2010B}\\
\hline
\hline GRS & &\\
\hline ${[[mn, m(2k-n+c), d_{z}\geq d / d_{x}\geq (d-c)]]}_{q}$ & $1< k < n < 2k+c \leq q^{m}$, & \cite{LaGuardia:2012}\\
\hline & $k=n-d+1$, $d > c+1$, $c, m \geq 1$ &\\
\hline
\end{tabular}
\end{center}
\end{table}

\clearpage

\clearpage

\begin{table}[bth!]
\begin{center}
\caption{Families of AQECC
\label{table2}}
\begin{tabular}{|c| c| c|}
\hline Code Family / ${[[n, k, d_{z}/ d_{x}]]}_{q}$ & Range of Parameters & Ref.\\
\hline RM & &\\
\hline ${[[2^{m}, k, 2^{m-r_{2}}\geq 2^{r_{1}+1}]]}_{2}$ & $0\leq r_{1} < r_{2} < m$,
$k=\displaystyle\sum_{j=r_{1}+1}^{r_{2}}
\left(
\begin{array}{c}
m\\
j\\
\end{array}
\right)$ & \cite{Sarvepalli:2009}\\
\hline
\hline Expanded GRM & &\\
\hline ${[[lq^{m}, \ l[k({\alpha}_{2})-k({\alpha}_{1})], \ d_{z}/d_{x}]]}_{p}$ &
$0\leq {\alpha}_1 \leq {\alpha}_{2}< m(q-1)$, $q=p^{l}$, $p$ prime, $l\geq 1$, &\\
\hline & $k(\alpha)=\displaystyle\sum_{i=0}^{m}
{(-1)}^{i}\left(
\begin{array}{c}
m\\
i\\
\end{array}
\right)
\left(
\begin{array}{c}
m+\alpha -iq\\
\alpha -iq\\
\end{array}
\right)$,  &\\
\hline & $d_{z}\geq d({\alpha}_{2})$, $d_{x}\geq d({\alpha}_{1}^{\perp})$, $d({\alpha}_{2})=(t+1)q^{u}$, &\\
\hline & $m(q-1)-{\alpha}_{2} =(q-1)u+t$, $0\leq t < q-1$, &\\
\hline & $d({\alpha}_{1}^{\perp})=(a+1)q^{b}$, ${\alpha}_{1}+1 =(q-1)b+a$, $0\leq a\leq q-1$ &\\
\hline
\hline MDS & &\\
\hline ${[[n, n - d_1 -d_{2} + 2, d_{z}/d_{x}]]}_{q}$ & $n=q-1$, $d_{x}=d_1 < d_{z}=d_{2}$ & \cite{Salah:2008}\\
\hline ${[[n, n - 2, 2/2]]}_{q}$ & $q$ prime power, $n\geq 3$ & \cite{Wang:2010}\\
\hline ${[[n, k-1, (n-k+1)/2]]}_{q}$ & $q\geq n> 3$, $1< k \leq n-2$ & \cite{Wang:2010}\\
\hline ${[[2^{m} + 2, 2, 2^{m}/2]]}_{2^{m}}$ & $m > 0$ integer  & \cite{Wang:2010}\\
\hline ${[[2^{m} + 2, 2^{m}-2, 4/2]]}_{2^{m}}$ & $m > 0$, $m\neq 2$ integer  & \cite{Wang:2010}\\
\hline ${[[n, j, d_{z}/d_{x}]]}_{q}$ & $n, k, j \in \mathbb{Z}$, $q\geq 5$, $n\leq q$,
$2 \leq k \leq n-3$, & \cite{Wang:2010}\\
\hline & $ j\leq n-k -2$, $\{ d_{z}, d_{x}\}=\{n-k-j+1, k+1\}$ &\\
\hline
\hline ${[[q+1, 2j, d_{z}/d_{x}]]}_{q}$ & $n, k, j \in \mathbb{Z}$, $q\geq 5$, $k\geq 2$,
$k+2j \leq q-1$, & \cite{Wang:2010}\\
\hline & $\{ d_{z}, d_{x}\}=\{q-k-2j+2, k+1\}$ &\\
\hline ${[[q+1, q-1-2s, (2s+1)/3]]}_{q}$ & $q=2^{m}\geq 4$, $s\leq q/2-1$ & \cite{Wang:2010}\\
\hline ${[[2^{m}+2, 2^{m}-4, 4/4]]}_{2^{m}}$ & $2^{m}\geq 4$ & \cite{Wang:2010}\\
\hline
\hline ${[[n, 2k-n+c, d_{z}\geq d / d_{x}\geq (d-c)]]}_{q}$ & $1< k < n < 2k+c \leq q$, & \cite{LaGuardia:2012}\\
\hline & $k=n-d+1$, $d > c+1$, $c\geq 1$ &\\
\hline
\hline Expanded Character & &\\
\hline ${[[t2^{m}, t[k(r_{2}) - k(r_1)], d_{z}/d_{x}]]}_{p}$ & $q=p^{t}$, $p$ odd prime, $t\geq 1$, & \\
\hline & $k(r)=\displaystyle\sum_{i=0}^{r}\left(\begin{array}{c}
m\\
i\\
\end{array}
\right)$, $d_{z}\geq 2^{m-r_{2}}$, $d_{x}\geq  2^{r_1 + 1}$ &\\
\hline
\hline QR & &\\
\hline ${[[p, 1, d_{z}/d_{x}]]}_{q}$ & $p$ prime, $p\equiv 1 \mod 4$, & \cite{Ketkar:2006}\\
\hline & $q=p_{1}^{t}$, $p\nmid p_{1}$, $q$ is a quadratic residue mod $p$, &\\
\hline & $d_{z}\geq \sqrt{p}$, $d_{x}\geq \sqrt{p}$ &\\
\hline
\hline ${[[p, 1, d_{z}/d_{x}]]}_{q}$ & $p$ prime, $p\equiv 3 \mod 4$, & \cite{Ketkar:2006}\\
\hline & $q=p_{1}^{t}$, $p\nmid p_{1}$, $q$ is a quadratic residue mod $p$, &\\
\hline & $d_{z}\geq d$, $d_{x}\geq d$, $d^{2}-d+1\geq p$ &\\
\hline
\hline Expanded QR & &\\
\hline ${[[tp, t, d_{z}/d_{x}]]}_{p_{*}}$ & $p$ prime, $p\equiv 1 \mod 4$, $q=p_{*}^{t}$, &\\
\hline & $t\geq 1$, $p\nmid p_{*}$, $q$ is a quadratic residue mod $p$, &\\
\hline & $d_{z}\geq \sqrt{p}$, $d_{x}\geq \sqrt{p}$ &\\
\hline
\hline ${[[tp, t, d_{z}/d_{x}]]}_{p_{*}}$ & $p$ prime, $p\equiv 3 \mod 4$, $q=p_{*}^{t}$, & \\
\hline & $t\geq 1$, $p\nmid p_{*}$, $q$ is a quadratic residue mod $p$, &\\
\hline & $d_{z}\geq d$, $d_{x}\geq d$, $d^{2}-d+1\geq p$ &\\
\hline
\hline Affine-Invariant & &\\
\hline ${[[tp^{m}, t(p^{m}-2 -2\frac{m}{t}), d_{z}/d_{x} ]]}_{p}$ & $q=p^{t}$, & \\
\hline & $p > 3$, $m > 2$, $d_{z}\geq d_{a}$, $d_{x}\geq d_{a}$, &\\
\hline & $d_{a}$ is given in Theorem~\ref{af-inv-lag} &\\
\hline
\hline Product code & &\\
\hline ${[[{(q - 1)}^{2}, \ (q - d_{1})(q - d_{3}) - (q - d_{2})(q - d_{4}),
\ d_{z}/d_{x}]]}_{q}$ & $2 \leq d_{1} \leq d_{2} < q-1$, $2 \leq d_3 \leq d_{4} < q-1$, &\cite{LaGuardia:2012I}\\
\hline & $d_{z}\geq \max \{ d_{1}d_{3}, \min\{ q-d_{2}, q-d_{4} \} \}$, &\\
\hline & $d_{x}\geq \min \{ d_{1}d_{3}, \min\{ q-d_{2}, q-d_{4} \} \}$ & \\
\hline
\end{tabular}
\end{center}
\end{table}

\clearpage


\small
\begin{thebibliography}{10}

\bibitem{Salah:2008}
S. A. Aly.
\newblock Asymmetric quantum BCH codes.
\newblock In {\em Proc. IEEE International Conference on
Computer Engineering and Systems (ICCES'08)}, pp. 157–-162, 2008.


\bibitem{Salah:2010}
S. A. Aly and A. Ashikhmin.
\newblock Nonbinary quantum cyclic and subsystem codes over
asymmetrically-decohered quantum channels.
\newblock e-print arXiv:quant-ph/1002.2966.


\bibitem{Ashikhmin:2000}
A. Ashikhmin, S. Litsyn, M. A. Tsfasman.
\newblock Asymptotically good quantum codes.
\newblock e-print arXiv:quant-ph/0006061.

\bibitem{Ashikhmin:2001}
A. Ashikhmin and E. Knill.
\newblock Non-binary quantum stabilizer codes.
\newblock {\em IEEE Trans. Inform. Theory}, 47(7):3065--3072, 2001.


\bibitem{Calderbank:1998}
A. R. Calderbank, E. M. Rains, P. W. Shor, and N. J. A. Sloane.
\newblock Quantum error correction via codes over $GF(4)$.
\newblock {\em IEEE Trans. Inform. Theory}, 44(4):1369--1387, 1998.


\bibitem{Ding:2000}
C. Ding, D. Kohel, S. Ling.
\newblock Elementary 2-group character codes.
\newblock {\em IEEE Trans. Inform. Theory}, 46(1):280--284, 2000.

\bibitem{Evans:2007}
Z. W. E. Evans, A. M. Stephens, J. H. Cole, and L. C. L.
Hollenberg.
\newblock Error correction optimisation in the presence of $x/z$ asymmetry.
\newblock e-print arXiv:quant-ph/0709.3875.

\bibitem{Ezerman:2010A}
M. F. Ezerman, S. Jitman, and S. Ling.
\newblock On asymmetric quantum MDS codes.
\newblock e-print arXiv:quant-ph/1006.1694.

\bibitem{Ezerman:2010}
M. F. Ezerman, S. Ling, and P. Sol\'e.
\newblock Additive asymmetric quantum codes.
\newblock e-print arXiv:quant-ph/1002.4088.

\bibitem{Ezerman:2010B}
M. F. Ezerman, S. Ling, O. Yemen and P. Sol\'e.
\newblock From skew-cyclic codes to asymmetric quantum codes.
\newblock e-print arXiv:quant-ph/1002.4088.


\bibitem{Grassl:2003}
M. Grassl, T. Beth, and M. R\"otteler.
\newblock On optimal quantum codes.
\newblock {\em Int. J. Quantum Inform.}, 2(1):757--766, 2004.

\bibitem{Grassl:1999}
M. Grassl, W. Geiselmann, and T. Beth.
\newblock Quantum Reed-Solomon codes.
\newblock {\em AAECC-13}, 1709:231--244, 1999.

\bibitem{Guenda:2009}
K. Guenda.
\newblock Quantum duadic and affine-invariant codes.
\newblock {\em Int. J. Quantum Inform.}, 7(1):373--384, 2009.


\bibitem{Huffman:2003}
W. C. Huffman and V. Pless.
\newblock {\em Fundamentals of Error Correcting Codes}.
\newblock Cambridge Univ. Press, 2003.

\bibitem{Ioffe:2007}
L. Ioffe and M. Mezard.
\newblock Asymmetric quantum error-correcting codes.
\newblock {\em Phys. Rev. A}, 75:032345(1--4), 2007.


\bibitem{Ketkar:2006}
A. Ketkar, A. Klappenecker, S. Kumar, and P. K. Sarvepalli.
\newblock Nonbinary stabilizer codes over finite fields.
\newblock {\em IEEE Trans. Inform. Theory}, 52(11):4892--4914, 2006.

\bibitem{LaGuardia:2009}
G. G. La Guardia.
\newblock Constructions of new families of nonbinary quantum codes.
\newblock {\em Phys. Rev. A}, 80(4):042331(1--11), 2009.

\bibitem{LaGuardia:2011}
G. G. La Guardia.
\newblock New families of asymmetric quantum BCH codes.
\newblock {\em Quantum Inform. Computation}, 11(3-4):239--252, 2011.

\bibitem{LaGuardia:2012}
G. G. La Guardia.
\newblock Asymmetric quantum Reed-Solomon and generalized Reed-Solomon codes.
\newblock {\em Quantum Inform. Processing}, 11:591--604, 2012.

\bibitem{LaGuardia:2011A}
G. G. La Guardia.
\newblock New quantum MDS codes.
\newblock {\em IEEE Trans. Inform. Theory}, 57(8):5551--5554, 2011.

\bibitem{LaGuardia:2012I}
G. G. La Guardia.
\newblock Asymmetric quantum product codes.
\newblock {\em Int. J. Quantum Inform.}, 10(1):1250005(1--11), 2012.


\bibitem{Lidl:1997}
R. Lidl and H. Niederreiter.
\newblock {\em Finite Fields}.
\newblock Cambridge Univ. Press, 1997.

\bibitem{Macwilliams:1977}
F. J. MacWilliams and N. J. A. Sloane.
\newblock {\em The Theory of Error-Correcting Codes}.
\newblock North-Holland, 1977.

\bibitem{Nielsen:2000}
M. A. Nielsen and I. L. Chuang.
\newblock {\em Quantum Computation and Quantum Information}.
\newblock Cambridge University Press, 2000.

\bibitem{Peterson:1972}
W. W. Peterson and E. J. Weldon Jr..
\newblock {\em Error-Correcting Codes}.
\newblock MIT Press, Cambridge, 1972.

\bibitem{Rains:1999}
E. M. Rains.
\newblock  Nonbinary quantum codes.
\newblock {\em IEEE Trans. Inform. Theory}, 45(6):1827--1832,
1999.


\bibitem{Sarvepalli:2008}
P. K. Sarvepalli, A. Klappenecker, and M. R\"otteler.
\newblock Asymmetric quantum LDPC codes.
\newblock In {\em Proc. Int. Symp. Inform. Theory (ISIT)}, pp. 6--11, 2008.

\bibitem{Sarvepalli:2009}
P. K. Sarvepalli, A. Klappenecker, and M. R\"otteler.
\newblock Asymmetric quantum codes: constructions, bounds and performance.
\newblock In {\em Proc. of the Royal Society A}, pp. 1645--1672, 2009.


\bibitem{Steane:1996}
A. M. Steane.
\newblock Simple quantum error correcting-codes.
\newblock {\em Phys. Rev. A}, 54:4741--4751, 1996.

\bibitem{Steane:1999}
A. M. Steane.
\newblock Enlargement of Calderbank-Shor-Steane quantum codes.
\newblock {\em IEEE Trans. Inform. Theory}, 45(7):2492--2495, 1999.

\bibitem{Stephens:2008}
A. M. Stephens, Z. W. E. Evans, S. J. Devitt, and L. C. L.
Hollenberg.
\newblock Asymmetric quantum error correction via code conversion.
\newblock {\em Phys. Rev. A}, 77:062335(1--5), 2008.

\bibitem{Wang:2010}
L. Wang, K. Feng, S. Ling, and C. Xing.
\newblock Asymmetric quantum codes: characterization and constructions.
\newblock {\em IEEE Trans. Inform. Theory}, 56(6):2938–-2945, 2010.

\end{thebibliography}
\end{document}